\documentclass[manuscript,screen,review=false]{acmart}
\renewcommand\footnotetextcopyrightpermission[1]{} % removes footnote with conference information in first column
\usepackage{algorithm}
\usepackage{algpseudocode}
\usepackage{bm} 
\usepackage{tabularx}
\usepackage{booktabs}

\AtBeginDocument{%
  }

\begin{document}

%% Title
\title{QR Sort: A Novel Non-Comparative Sorting Algorithm}

%% Authors
\author{Randolph T. Bushman}
\email{rbushma1@jhu.edu}
\orcid{0009-0003-3685-8850} 
\affiliation{%
  \institution{Johns Hopkins University Applied Physics Laboratory}
  \city{Laurel}
  \state{Maryland}
  \country{USA}
}

\author{Tanya M. Tebcherani}
\email{tebcherani.tanya@gmail.com}
\orcid{0000-0002-8357-9040}
\affiliation{%
  \institution{Johns Hopkins University Applied Physics Laboratory}
  \city{Laurel}
  \state{Maryland}
  \country{USA}
  % \institution{Johns Hopkins University Department of Mechanical Engineering}
  % \city{Baltimore}
  % \state{Maryland}
  % \country{USA}
}

\author{Alhassan S. Yasin}
\email{ayasin1@jhu.edu}
\orcid{0009-0001-8033-9850}
\affiliation{%
  \institution{Johns Hopkins University}
  \city{Baltimore}
  \state{Maryland}
  \country{USA}
}

%% Abstract
\begin{abstract}
In this paper, we introduce and prove QR Sort, a novel non-comparative integer sorting algorithm. This algorithm uses principles derived from the Quotient-Remainder Theorem and Counting Sort subroutines to sort input sequences stably. QR Sort exhibits the general time and space complexity $\mathcal{O}(n+d+\frac{m}{d})$, where $n$ denotes the input sequence length, $d$ denotes a predetermined positive integer, and $m$ denotes the range of input sequence values plus 1. Setting $d = \sqrt{m}$ minimizes time and space to $\mathcal{O}(n + \sqrt{m})$, resulting in linear time and space $\mathcal{O}(n)$ when $m \leq \mathcal{O}(n^2)$. We provide implementation optimizations for minimizing the time and space complexity, runtime, and number of computations expended by QR Sort, showcasing its adaptability.  Our results reveal that QR Sort frequently outperforms established algorithms and serves as a reliable sorting algorithm for input sequences that exhibit large $m$ relative to $n$. 

\end{abstract}

% <concept_id>10003752.10003809.10010031.10010033</concept_id>
% <concept_desc>Theory of computation~Sorting and searching</concept_desc>

% \ccsdesc[500]{Theory of computation~Sorting and searching}

%% Keywords

% Paper __must first be submitted / accepted
% \received{}
% \received[revised]{}
% \received[accepted]{}

\maketitle
\fancyfoot{}
\thispagestyle{empty}

\section{Introduction}
In computer science, sorting algorithms arrange sequences of values into a specific order to improve data processing and optimize performance across various computational tasks \cite{cormen_et_al}. Common operations that benefit from sorting include binary search for quick lookup \cite{lin_binary_search}, median finding for statistical analysis \cite{schonhage_median_finding}, and prioritization for organizing tasks~\cite{javed_requirement_prioritization}. 

Comparison-based algorithms represent a common sorting paradigm that leverages a specified abstract comparison operator, such as ``less than or equal to," to order elements in a given sequence \cite{dey_comparison_based_sorting_complexity}. These algorithms possess the proven lower bound time complexity $\mathcal{O}(n\log{n})$, where $n$ denotes the length of the input sequence $S$ \cite{morris_sorting_theorems}. Time complexity helps assess the general performance of algorithms by bounding the computational time relative to the given input size. Likewise, space complexity helps assess algorithm performance by bounding the memory required relative to input size \cite{cormen_et_al}. Comparison-based algorithms may each possess different space complexities. Merge Sort and Quicksort \cite{cormen_et_al} serve as classic examples of comparison-based sorting algorithms. We provide the time and space complexities of these algorithms in Table \ref{tab:complexity_comparison}.

Non-comparative integer sorting algorithms classify a distinct group of sorting methods with no proven lower-bound time complexity greater than $\mathcal{O}(n)$. These algorithms associate input values with integer keys and distribute them into ordered bins to enable efficient derivation of the final sorted order, circumventing the comparison-based lower bound~\cite{cormen_et_al}.

This paper introduces QR Sort, a novel non-comparative integer sorting algorithm that divides each input element by a user-specified divisor and uses the acquired quotient and remainder values as sorting keys. We proved QR Sort qualifies as a stable sorting algorithm that ensures elements with equal keys maintain their initial relative order after sorting \cite{tang_sorting_introduction}. We also provided implementation optimizations for the time and space complexity, runtime, and number of computations expended by QR Sort. Our results reveal that QR Sort frequently outperforms established algorithms and serves as a reliable sorting algorithm for input sequences that exhibit larger input sequence element
ranges.

\section{QR Sort}\label{qr_sort_section}
We developed QR Sort, a stable integer sorting algorithm derived from the Quotient-Remainder Theorem, to improve performance and efficiency compared to existing sorting algorithms. In this section, we provide a description and proof of the QR Sort algorithm.

\subsection{General Algorithm}\label{qr_sort_section_general_algorithm}
Let $S = \{s_1, s_2, \ldots, s_n\}$ represent a sequence of $n$ integers and let the integer $m$ denote the range of values in $S$ plus one, such that $m = \max(S) - \min(S) + 1$. Let $f(S, K)$ denote a stable sorting function that sorts each $s_i$ in $S$ by its corresponding key $k_i$ in the sequence $K$. QR Sort employs sequences of normalized remainder and quotient keys $R$ and $Q$, detailed in Eq. \eqref{equation_k_r} and Eq. \eqref{equation_k_q}, where $d$ equals some predetermined positive integer divisor. QR Sort aligns the keys $r_i$ and $q_i$ with $s_i$ such that any changes made to the order of elements in $S$ reflect the relative ordering of $R$ and $Q$.

\begin{figure*}
    \centering
    \includegraphics[width=\linewidth]{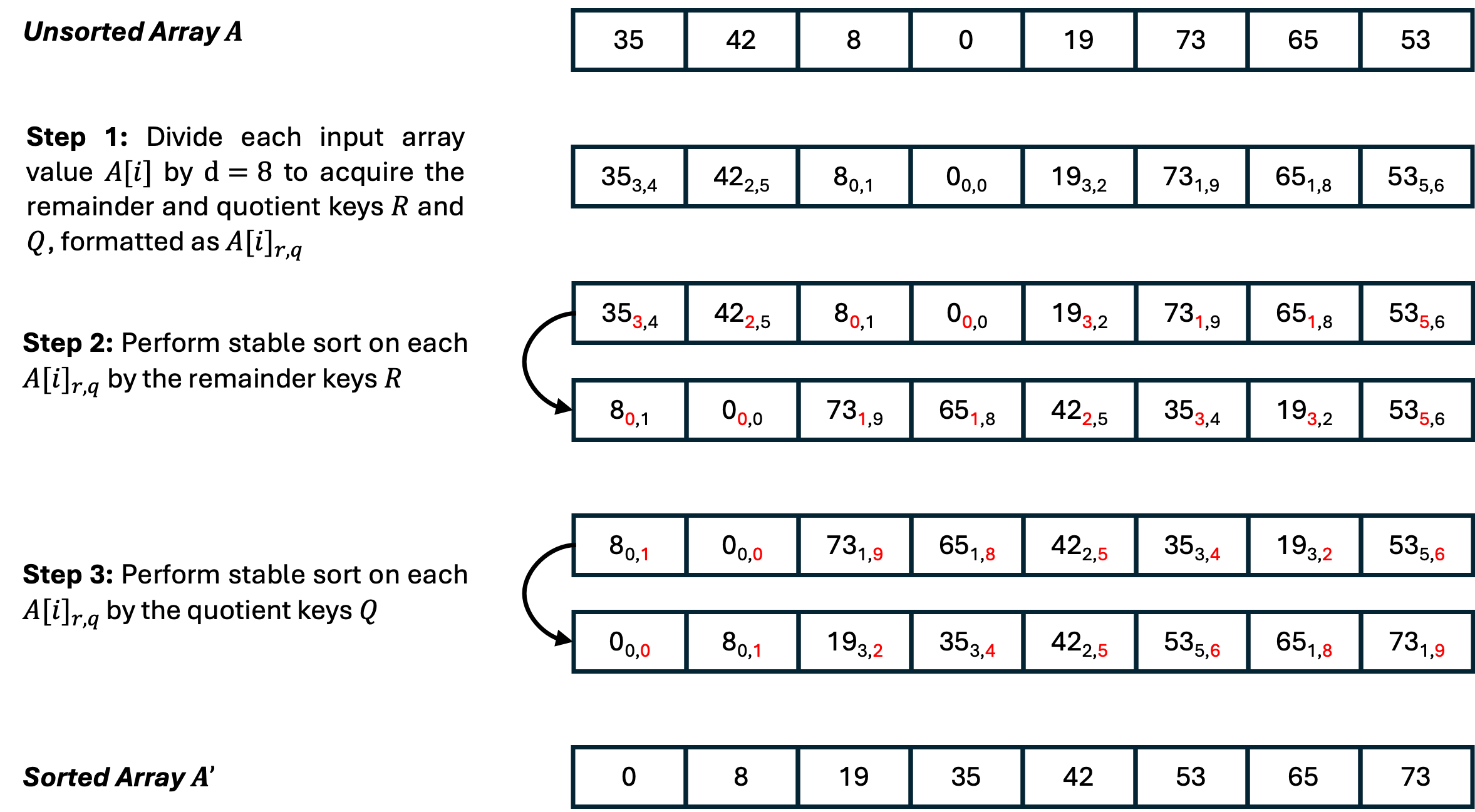}
    \caption{Visualization of the sorting procedure of the QR Sort Algorithm given an integer array input with length $n = 8$. For the divisor $d$, we selected the optimal divisor described in Proposition \ref{OptimalDivisorTheorem_RQ}, or $d = \sqrt{m} = \sqrt{73-0+1} = 8$. }
    \label{fig:qr_fig}
    \Description{}
\end{figure*}

\begin{equation}\label{equation_k_r}
r_i = (s_i - \min(S)) \bmod d 
\end{equation}

\begin{equation}\label{equation_k_q}
q_i = \lfloor \frac{s_i-\min(S)}{d} \rfloor
\end{equation}

QR Sort conducts two stable sorting routines such that the final sorted sequence $S' = f(f(S, R), Q)$. Here, QR Sort first sorts $S$ by the remainder keys $R$, then sorts the resulting sequence by the quotient keys $Q$ to obtain the final sorted sequence $S'$. We note that QR Sort reduces the key space so effectively that the combined work of its two sorting subroutines often results in fewer computations performed than algorithms that require a single sorting pass. We illustrated QR Sort sorting an integer array in Figure \ref{fig:qr_fig}. If $\lfloor\frac{m}{d}\rfloor = 0$, it indicates every $q_i$ equals zero and renders the stable sort by $Q$ redundant and unnecessary. Thus, the final sorted sequence simplifies to $S' = f(S, R)$.

\subsection{Proof of Sorting and Stability}\label{qr_sort_section_proof}
QR Sort constitutes a stable sorting algorithm that orders the elements in an input sequence $S$ to produce a final sorted sequence $S'$. First, we define a stable sorting algorithm in Definition \ref{stable_def}. 

\indent Specifically, QR Sort produces the final sorted sequence $S' = f(S_R, Q)$, where $S_R = f(S, R)$. QR Sort computes the remainder and quotient keys $r_i$ and $q_i$ in sequences $R$ and $Q$, respectively, via Eq. \eqref{equation_k_r} and Eq. \eqref{equation_k_q}, where the divisor $d$ denotes some positive integer. The proofs for the guaranteed sorting and stability of QR Sort follow. 

\begin{lemma}\label{Lemma}
Let $S$ denote a sequence of integers, and for each $s_i$ in $S$, define the remainder $r_i$ and quotient $q_i$ using Eqs.~\eqref{equation_k_r}~and~\eqref{equation_k_q}, where $d$ denotes some positive integer divisor. If $s_i < s_j$ and $r_i \geq r_j$, then $q_i < q_j$.
\end{lemma}

\begin{proof}
\noindent Let $s_i < s_j$. We subtract $\min(S)$ from both sides, floor divide by $d$ on the results, and use Eq. \eqref{equation_k_q} to derive Eq.~\eqref{eq_qi_leq_qj}. 

\begin{equation}\label{eq_qi_leq_qj}
 s_i < s_j \implies \lfloor \frac{s_i-\min(S)}{d} \rfloor \leq \lfloor \frac{s_j-\min(S)}{d} \rfloor \implies q_i \leq q_j   
\end{equation}

The transition from $<$ to $\leq$ in Eq.~\eqref{eq_qi_leq_qj} follows from the floor function, as it truncates the result of the division operation. Now, let $r_i \geq r_j$. We apply the Quotient-Remainder Theorem \cite{kwong_quotient_remainder_theorem} in Eq. \eqref{eq_quotient_remainder_thm} to derive the expression in Eq. \eqref{eq_ri_lt_rj_init}:

\begin{equation}\label{eq_quotient_remainder_thm}
s_i = dq_i + r_i
\end{equation}

\begin{equation}\label{eq_ri_lt_rj_init}
s_i < s_j \implies dq_i + r_i < dq_j + r_j 
\end{equation}

Next, assume for contradiction that $q_i = q_j$. Then, Eq. \eqref{eq_ri_lt_rj_init} becomes:

\begin{equation}\label{eq_ri_lt_rj}
dq_i + r_i < dq_j + r_j \implies r_i < r_j
\end{equation}

The result, $r_i < r_j$, contradicts the given condition $r_i \geq r_j$ and thus invalidates our assumption that $q_i = q_j$. Therefore, using our finding in Eq. \eqref{eq_qi_leq_qj} that $q_i \leq q_j$, we conclude:

\begin{equation}\label{eq_qi_lt_qj}
    q_i \leq q_j\text{ and } q_i \neq q_j \implies q_i < q_j
\end{equation}
\end{proof}

\begin{definition}\label{stable_def} 
Let $s_i$ and $s_j$ denote any two elements in an input sequence $S$, where $s_i \longrightarrow s_j$ signifies that $s_i$ precedes $s_j$. The stable sorting function $f(S,K)$ sorts each $s_i$ in $S$ based on the corresponding $k_i$ in $K$ and outputs the sorted sequence $S'$ such that if $k_i < k_j$, then $s_i \longrightarrow s_j$ in $S'$. The inherent stability of $f(S,K)$ also infers that if $s_i \longrightarrow s_j$ in $S$ and $k_i = k_j$, then $s_i \longrightarrow s_j$ in $S'$. 
\end{definition}

\begin{theorem}\label{SortTheorem}
For a sequence $S$ of $n$ integers, let $R$ and $Q$ denote the sequences of normalized remainder and quotient keys, respectively. Let each remainder $r_i$ and quotient $q_i$ correspond to each element $s_i$ in $S$ with Eqs.~\eqref{equation_k_r} and \eqref{equation_k_q}, where $d$ denotes a positive integer divisor. For each $s_i$ in $S$, QR Sort guarantees that if $s_i < s_j$, then $s_i \longrightarrow s_j$ in the resulting sorted sequence $S'$.
\end{theorem}

\begin{proof}
\noindent Let $S_R = f(S,R)$, as defined by Def. \ref{stable_def}. Also, let $s_i < s_j$. Then, by Eq. \eqref{eq_qi_leq_qj}, $q_i \leq q_j$. This leaves two possible scenarios for the composition of the remainder keys: \\

1. Suppose $r_i \geq r_j$. By Lemma \ref{Lemma}, the condition $q_i < q_j$ holds true, which leads $S' = f(S_R,Q) \text{ to produce } s_i \longrightarrow s_j$. \\

2. Suppose $r_i < r_j$. Thus, $S_R$ produces $s_i \longrightarrow s_j$. Since $q_i \leq q_j$, $S' = f(S_R, Q)$ preserves $s_i \longrightarrow s_j$.
\end{proof}

\begin{theorem}\label{StableTheorem}
For a sequence $S$ of $n$ integers, let $R$ and $Q$ denote the sequences of normalized remainder and quotient keys, respectively. Let each remainder $r_i$ and quotient $q_i$ correspond to each element $s_i$ in $S$ with Eqs.~\eqref{equation_k_r} and \eqref{equation_k_q}, where $d$ denotes a positive integer divisor. QR Sort guarantees stability, such that if $s_i \longrightarrow s_j$ in the input sequence $S$ and $s_i = s_j$, then $s_i \longrightarrow s_j$ in the resulting sequence $S'$.
\end{theorem}

\begin{proof}
Let $S_R = f(S,R)$, as defined by Def. \ref{stable_def}. Also, let $s_i = s_j$ and $s_i \longrightarrow s_j$.\\

Since $s_i = s_j$, it follows that $r_i = r_j$ and $q_i = q_j$. Since $s_i \longrightarrow s_j$ in $S$, $f(S, R) \text{ preserves } s_i \longrightarrow s_j \text{ in } S_R$ and $f(S_R, Q) \text{ preserves } s_i \longrightarrow s_j \text{ in } S'$ due to the stability of $f$ as defined in Def. \ref{stable_def}.
\end{proof}

\section{Implementation \& Optimizations}\label{implementation_section}
In this section, we present recommendations for implementing QR Sort using Counting Sort \cite{cormen_et_al} subroutines and an optimal divisor value $d$, with additional optimizations for QR Sort depending on the composition of the input sequence $S$. 

\subsection{General Sorting Algorithm Recommendation for QR Sort Subroutines}\label{implementation_section_general_implementation}

QR Sort requires the use of a stable sorting algorithm to sort the remainder and quotient keys $R$ and $Q$. We recommend using Counting Sort \cite{cormen_et_al} Figure \ref{fig:counting_fig} as it constitutes a stable sorting algorithm with an attainable linear time complexity. Here, let Counting Sort implement the stable sorting function $f(S,K)$ utilized by QR Sort, where $K$ equals $R$ or $Q$, and $S$ sorts by $K$. Counting Sort creates a series of $m + 1$ ordered bins labeled from zero to $m$, where $m = \max(K) - \min(K) + 1$. The bin labeled $k_i - \min(K)$ tallies the frequency of each element $k_i$. Counting Sort then decrements the tally of the first bin by 1 and performs a cumulative sum such that the value in each bin increments by the sum of the values in the preceding bins. Next, Counting Sort iterates backwards through $K$ and places each element $k_i$ in its final sorted location using the value in the bin labeled $k_i - \min(K)$ as its index. Immediately after sorting an element $k_i$, Counting Sort decrements the value in the bin labeled $k_i - \min(K)$ before sorting the next element $k_{i-1}$ to compensate for duplicate values in $K$. Finally, given that $K$ sorts $S$, the final sorted order of $K$ also reflects the final sorted order of $S$. 

\begin{figure*}
    \centering
    \includegraphics[width=\linewidth]{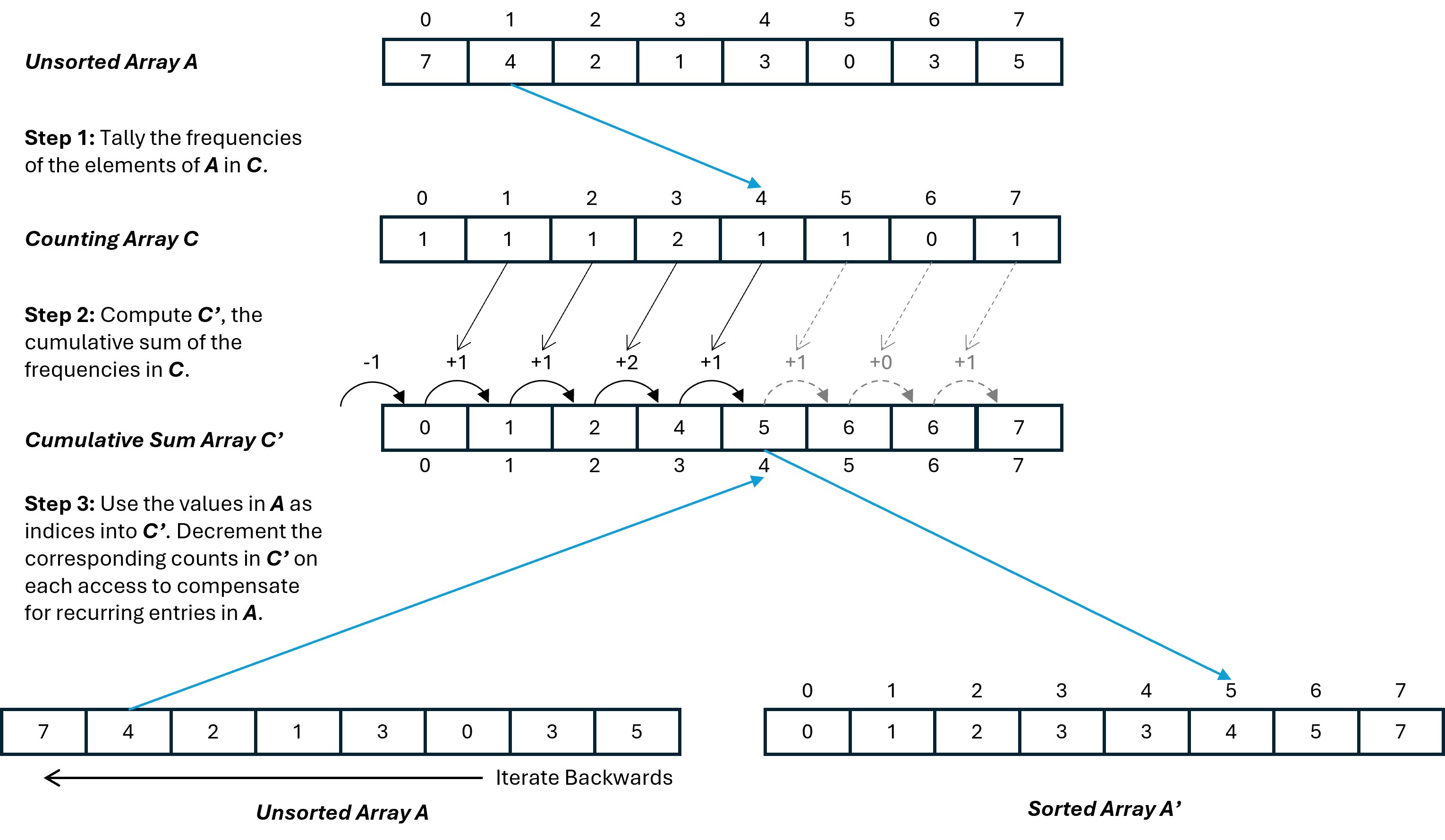}    
    \caption{Visualization of the sorting procedure of the Counting Sort algorithm given an integer array input of size $n = 8$ and a range of~$7$~($m = 8$).}
    \label{fig:counting_fig}
    \Description{}
\end{figure*}

Counting Sort exhibits the time and space complexity $\mathcal{O}(n+m)$ \cite{cormen_et_al}. When sorting by the remainder keys in QR Sort with Counting Sort, the time and space complexity equals $\mathcal{O}(n+d)$. Similarly, when sorting the quotient keys using Counting Sort, the time and space complexity equals $\mathcal{O}(n+\frac{m}{d})$. Thus, the overall time and space complexity exhibited by this implementation of QR Sort equals $\mathcal{O}(n+d+\frac{m}{d})$. 

We present pseudo-code for an implementation of Counting Sort that sorts input sequences based on representative integer keys in Algorithm \ref{appendix_pseudo-code_counting_sort}. We also present pseudo-code for the described implementation of QR Sort using Counting Sort subroutines in Algorithm \ref{appendix_pseudo-code_qr_sort}.

\subsection{Divisor $d$ Selection for Minimizing Time Complexity}\label{implementation_section_recommended_divisor}

We present an approach for selecting a divisor $d$ in Theorem \ref{OptimalDivisorTheorem_RQ} that minimizes the time complexity of QR Sort. We recommend this divisor for general applications of QR Sort, but acknowledge that other divisor selections may optimize other applications. 

\begin{proposition}\label{OptimalDivisorTheorem_RQ}  
Let $S$ denote a sequence of $n$ integers, and let $m = \max(S) - \min(S) + 1 > 0$ denote the range of values in $S$ plus one. The integer divisor $d = \sqrt{m}$ minimizes the time and space complexity of QR Sort, reducing it from $\mathcal{O}(n+d+\frac{m}{d})$ to $\mathcal{O}(n+\sqrt{m})$.
\end{proposition}

\begin{proof}
Let $f(d) = n + d + \frac{m}{d}$.  We first differentiate $f$ with respect to $d$ to derive Eq. \eqref{eq_derivative}.

\begin{equation}\label{eq_derivative}
f'(d) = 1 - \frac{m}{d^2}
\end{equation}

To minimize $f$, we first set $f'$ to zero and then solve for $d$ in Eq. \eqref{eq_optimal_d}.

\begin{equation}\label{eq_optimal_d}
0 = 1 - \frac{m}{d^2} \implies d = \pm \sqrt{m}
\end{equation}

Since $m > 0$, Eq. \eqref{eq_optimal_d} simplifies to $d = +\sqrt{m}$. To verify this value minimizes $f$, we calculate the second derivative $f''$ in Eq. \eqref{eq_second_derivative}.

\begin{equation}\label{eq_second_derivative} 
f''(d) = \frac{2m}{d^3} > 0 
\end{equation}

Since $f'' > 0$, $f''(\sqrt{m})$ must return a positive concave-up solution, confirming that $f$ reaches a local minimum at this point. Substituting $\sqrt{m}$ into the time complexity expression $\mathcal{O}(n+d+\frac{m}{d})$ reduces it to $\mathcal{O}(n+\sqrt{m})$

\end{proof}

\subsection{Bypassing Quotient Sort to Execute in Linear Time and Space}\label{which_sort}
As described in Section \ref{qr_sort_section_general_algorithm}, QR Sort sorts an input sequence $S$ by first sorting remainder keys $R$, followed by sorting quotient keys $Q$. However, when $n \propto m$ and the divisor $d = m + 1$, QR Sort executes in $\mathcal{O}(n)$ time and space and bypasses the quotient sort, thereby reducing significant processing time. 

\begin{corollary}\label{theorem_4_bypass_quotiet}
Let $S$ denote a sequence of $n$ integers, and let $m$ denote the integer range of values in $S$ plus one such that $m = \max(S) - \min(S) + 1$. Let $d$ denote some positive integer divisor. Recall from Section \ref{implementation_section_general_implementation} the general time complexity of QR Sort $\mathcal{O}(n+d+\frac{m}{d})$. When $d = m + 1$ and $n \propto m$, QR Sort executes in $\mathcal{O}(n)$ time and space and bypasses the quotient sort.
\end{corollary}

\begin{proof}
Recall that in Section \ref{qr_sort_section_general_algorithm}, we stated that QR Sort bypasses the final quotient sort when $\lfloor\frac{m}{d}\rfloor = 0$. The smallest possible $d$ that meets this criterion equals $m+1$, as it represents the smallest integer greater than $m$. When $d=m+1$, the the final time complexity evaluates to $\mathcal{O}(n+m)$. Next, let $n \propto m$, which implies $m = c \cdot n$, where $c$ denotes some constant. Substituting this into $\mathcal{O}(n + m)$ yields Eq. \eqref{n_m_proprtional}.

\begin{equation}\label{n_m_proprtional}
\mathcal{O}(n + m) = \mathcal{O}(n + c \cdot n) = \mathcal{O}(n)
\end{equation}
\end{proof}

\subsection{Subtraction-Free Sorting}\label{implementation_section_subtract_free_sorting}
Under the recommended implementation, QR Sort subtracts $\min(S)$ from each $s_i$ before each division to compute the remainder and quotient keys. This strategy minimizes the number of counting bins and prevents potential errors caused by negative indexing when $S$ contains negative elements. If the input consists of non-negative integers, QR Sort may bypass this step. However, this modification may degrade computational performance when sorting sequences with a substantial $\min(S)$ since QR Sort then generates more bins than necessary compared to employing the original implementation.

\subsection{Bitwise Operation Substitution}\label{implementation_section_bitwise_operations}
Computing the remainder and quotient keys $R$ and $Q$ when $S$ generally comprises large elements involves computationally expensive modulo and division operations. By assigning $d$ as a power of two where $d=2^c$, QR Sort may leverage bitwise operations to compute $R$ and $Q$. Employing bitwise arithmetic enhances algorithm runtime, as modern processor architectures offer optimizations for these operations \cite{yordzhev_bitwise_operations}.

We compute $R$ and $Q$ using Eq. \eqref{equation_rem} and Eq. \eqref{equation_quot}, respectively \cite{yordzhev_bitwise_operations, brown_binary_arithmetic}. Here, $\gg$ represents the right-bit shift operator, $\&$ represents the bitwise AND operator, and $c$ denotes the exponent for $d = 2^c$, specifying the number of bits used in the operations to compute the remainders and quotients.

\begin{equation}\label{equation_rem}
r_i = (s_i - \min(S)) ~ \& ~ (2^c-1)
\end{equation}

\begin{equation}\label{equation_quot}
q_i = (s_i - \min(S)) \gg c
\end{equation}

\section{Analyses}\label{analyses_section}
In this section, we outline a comparative analysis of the computational performance of QR Sort against other sorting algorithms.

\subsection{SortTester\_C Experiment Setup}\label{analyses_section_setup}
We developed SortTester\_C, a program designed to evaluate the computation performance of QR Sort compared to other sorting algorithms when sorting integer arrays of various sizes Figure \ref{fig:block_diagram} \cite{bushman_sort_tester}.

\begin{figure*}\label{block_diagram}
    \centering
    \includegraphics[width=.97\linewidth]{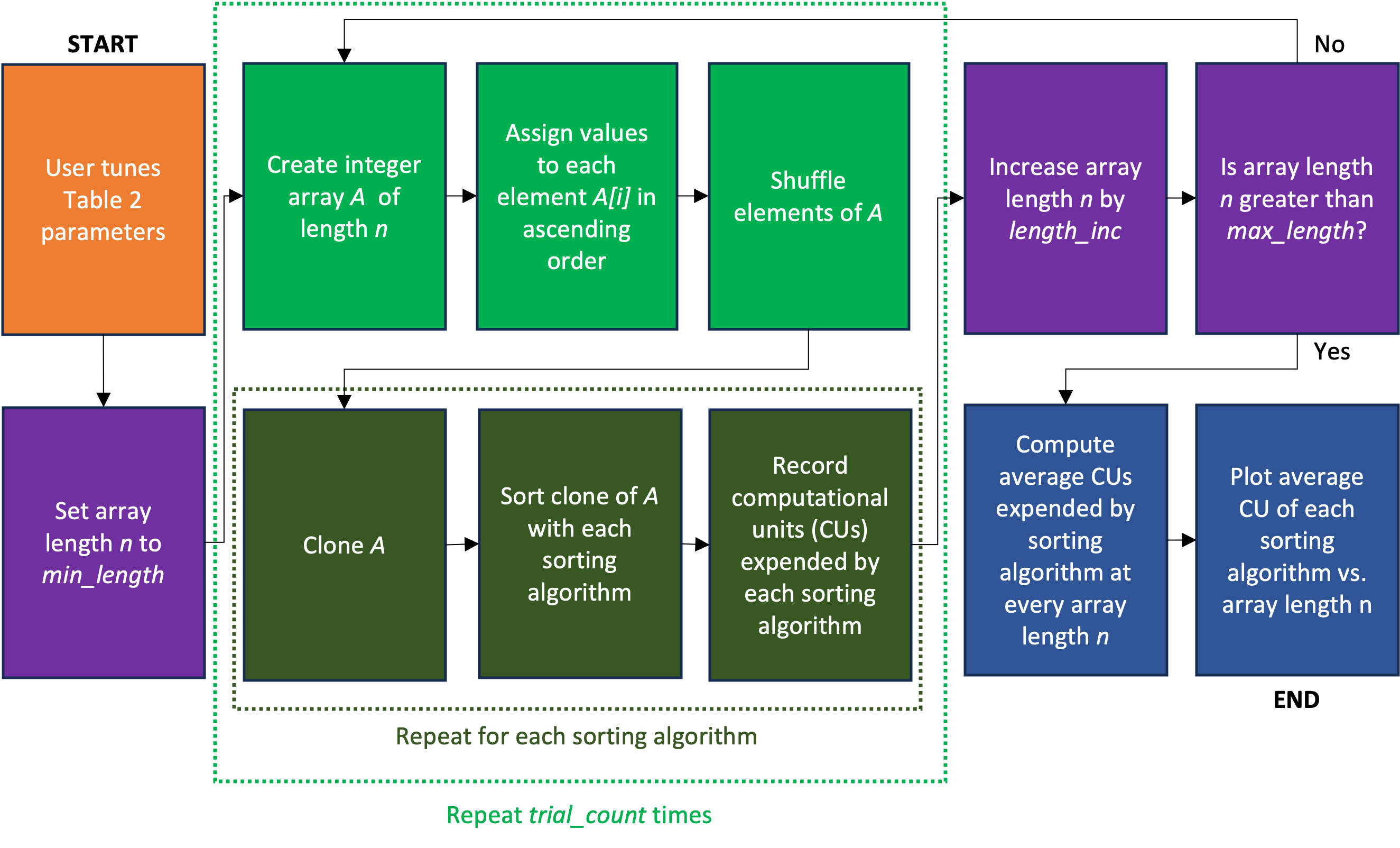}
    \caption{Block diagram of one SortTester\_C experiment used to evaluate the performance of multiple sorting algorithms in terms of computational units (CU).}
    \label{fig:block_diagram}
    \Description{}
\end{figure*}

\begin{table*}
  \caption{Complexity Comparison of Sorting Algorithms}
  \label{tab:complexity_comparison}
  \begin{tabularx}{\textwidth}{l l l l l l} % Changed alignment to left
    \toprule
    Algorithm & Type & Best Case Time & Worst Case Time & Best Case Space & Worst Case Space\\
    \midrule
    Merge Sort & Comparison-Based & $\mathcal{O}(n\log{n})$ & $\mathcal{O}(n\log{n})$ & $\mathcal{O}(n)$ & $\mathcal{O}(n)$\\
    Quicksort & Comparison-Based & $\mathcal{O}(n\log{n})$ & $\mathcal{O}(n^2)$ & $\mathcal{O}(\log{n})$ & $\mathcal{O}(n)$\\
    Counting Sort & Non-Comparative & $\mathcal{O}(n)$ & $\mathcal{O}(n + m)$ & $\mathcal{O}(n)$ & $\mathcal{O}(n + m)$\\
    Radix Sort & Non-Comparative & $\mathcal{O}(n)$ & $\mathcal{O}(n\log{_n}{m})$ & $\mathcal{O}(n)$ & $\mathcal{O}(n)$\\
    QR Sort & Non-Comparative & $\mathcal{O}(n)$ & $\mathcal{O}(n + \sqrt{m})$ & $\mathcal{O}(n)$ & $\mathcal{O}(n + \sqrt{m})$\\
    \bottomrule
  \end{tabularx}
\end{table*}

The user first tunes the parameters of the program, described in Table \ref{tab:parameters}. SortTester\_C then begins by creating an integer array $A$ of size $n$, where $n$ initially equals \texttt{min\_length}. It assigns a uniform value in ascending order to each element $A[i]$ from \texttt{min\_value} to \texttt{max\_value}. SortTester\_C conducts an experiment composed of \texttt{trial\_count} sequential trials where within each trial, it shuffles the elements in $A$ using the Fisher-Yates Shuffle Algorithm \cite{fisher_yates_shuffle}. After shuffling, the program sorts $A$ with a set of predetermined sorting algorithms and records the computational units expended by each algorithm. Here, computational units encompass array accesses, comparisons, division, modulo, and bitwise operations. We assigned a weight of fifteen computational units to the division and modulo operations as we recognized the increased computational demand for both. SortTester\_C then increments the array length by \texttt{length\_inc} to begin the next experiment. This process repeats until the array length surpasses \texttt{max\_length}.

\begin{table*}
  \caption{Tunable SortTester\_C Parameters}
  \label{tab:parameters}
  \begin{tabularx}{\textwidth}{l p{6cm} p{3.5cm}}
    \toprule
    Parameter Name & Description & Experiment Values \\
    \midrule
    \texttt{min\_length} & Minimum length of $A$ & 10,000 \\
    \texttt{max\_length} & Maximum length of $A$ & 1,000,000 \\
    \texttt{length\_inc} & Value by which the length of $A$ increments after each experiment & 10,000 \\
    \texttt{min\_value} & Smallest value in the array $A$ & 0 \\
    \texttt{max\_value} & Largest value in the array $A$ & 50,000; 500,000; 50,000,000; 500,000,000 \\
    \texttt{trial\_count} & Number of trials conducted in each SortTester\_C experiment & 10 \\
    \bottomrule
  \end{tabularx}
\end{table*}

To visualize the experimental results, SortTester\_C generates plots of the natural logarithm for the average computational units expended by each sorting algorithm against array length $n$. Using a logarithmic scale reveals nuanced patterns obscured by the extensive range of computational units employed by these algorithms.

\subsection{Common Sorting Algorithms}\label{analyses_section_common_algorithms}
We compared the computational performance of QR Sort to four other common sorting algorithms: Merge Sort, Quicksort, Counting Sort, and LSD Radix Sort. Merge Sort and Quicksort both constitute comparison-based divide-and-concur sorting algorithms. Both algorithms perform recursion to split the inputs; however, Merge Sort merges auxiliary array segments while Quicksort performs in-place partitioning. Counting Sort and LSD Radix Sort both constitute non-comparative sorting algorithms that leverage auxiliary bins to sort input sequences. Counting Sort sorts input sequences by element frequencies, whereas Radix Sort sorts by the digits of each sequence element. We present an in-depth explanation of the Counting Sort algorithm in Section \ref{implementation_section_general_implementation}. Additionally, we list the associated time and space complexities of these algorithms in Table \ref{tab:complexity_comparison} \cite{cormen_et_al}. 

\subsection{Comparison of QR Sort to Common Sorting Algorithms}\label{analyses_section_comparison_to_other_algorithms}
To compare the computational performance of QR Sort to common sorting algorithms Merge Sort, Quicksort, Counting Sort, and LSD Radix Sort, we ran four SortTester\_C experiments using the values in Table \ref{tab:parameters}. The four experiments held each input parameter constant, except for \texttt{max\_value} which increased for each experiment. We chose these values to account for randomness and to explore general trends across a large span of array lengths and element ranges. For each experiment, we used Merge Sort, Quicksort, Counting Sort, Radix Sort, and QR Sort to sort the input arrays. For Radix Sort, we set the base of the number system to $n$, as it simplifies its time complexity to $\mathcal{O}(n\log{_n}{m})$ \cite{mit_radix_sort}. 

As described in Section \ref{analyses_section_setup}, SortTester\_C generated a plot of the natural logarithm of the average computational units expended by each sorting algorithm by $n$ for all the four experiments, where each experiment used a different value for \texttt{max\_value}. We aimed to identify the sorting algorithm that expended the least number of computational ~units.

\section{Results \& Discussion}\label{results_and_discussion_section}
We compared the computational efficiency of QR Sort against Merge Sort, Quicksort, Counting Sort, and Radix Sort \cite{cormen_et_al}. In all four of our experiments, QR Sort expended fewer computational units than Merge Sort, Quicksort, and Radix Sort. Counting Sort, however, expended fewer computational units than QR Sort in every trial when $m=50,000$ and $m=500,000$. With $m=5,000,000$, QR Sort outperformed Counting Sort for smaller $n$, though Counting Sort expended fewer computational units for $n \geq 370,000$. Finally, at $m=50,000,000$ QR Sort outperformed Counting Sort for all tested $n$.

Our results indicate that the non-comparative nature of QR Sort enables it to outperform the comparison-based algorithms Merge Sort and Quicksort (Table \ref{tab:complexity_comparison}) for small $m$. Comparison-based algorithms depend solely on $n$, while non-comparative algorithms depend on both $n$ and $m$. For small $m$, non-comparative algorithms achieve near-linear sorting times and outperform comparison-based ones. However, as $m$ increases, non-comparative performance degrades, whereas comparison-based algorithms, bound by $\mathcal{O}(n\log{n})$, maintain consistent performance as they exclude $m$ entirely \cite{cormen_et_al}. For this reason, we expect both Merge Sort and Quicksort to eventually surpass all three non-comparative algorithms as $m$ grows.

Although both constitute non-comparative algorithms, we also observed that QR Sort outperformed Radix Sort in our experiments. We expect this trend to continue for moderately large $m$, as QR Sort executes fewer intensive division and modulo operations to obtain its keys compared to Radix Sort \cite{cormen_et_al, bushman_sort_tester} as shown in Section \ref{analyses_section_comparison_to_other_algorithms}. However, this expectation falls short in theory for very large $m$. While QR Sort exhibits the time complexity of $\mathcal{O}(n + \sqrt{m})$, Radix Sort operates in $\mathcal{O}(n \log_n m)$ time as shown in Table \ref{tab:complexity_comparison}. Thus, when $m > \mathcal{O}(n^2)$, QR Sort fails to achieve linear sorting. In contrast, when $m = \mathcal{O}(n^c)$ for any constant $c \geq 1$, Radix Sort still operates in linear time \cite{mit_radix_sort}.

Compared to both Radix Sort and QR Sort, Counting Sort operates as the most efficient algorithm for smaller $m$ as it performs fewer and simpler operations limited to addition, subtraction, and array accesses. However, the performance of Counting Sort degrades linearly as $m$ grows, as it exhibits the time complexity $\mathcal{O}(n+m)$. In contrast, the performance of QR Sort degrades at a slower rate of $\mathcal{O}(\sqrt{m})$. In our experiments, Counting Sort outperforms QR Sort for smaller $m$ Figures \hyperref[graph_figure]{4A} and \hyperref[graph_figure]{4B}, while QR Sort outperforms Counting Sort for larger $m$ Figures \hyperref[graph_figure]{4C} and \hyperref[graph_figure]{4D}. Therefore, although Counting Sort sometimes outperforms QR Sort, it proves less reliable without prior knowledge of $m$ in the input sequence.

\begin{figure*}
    \centering
    \includegraphics[width=\linewidth]{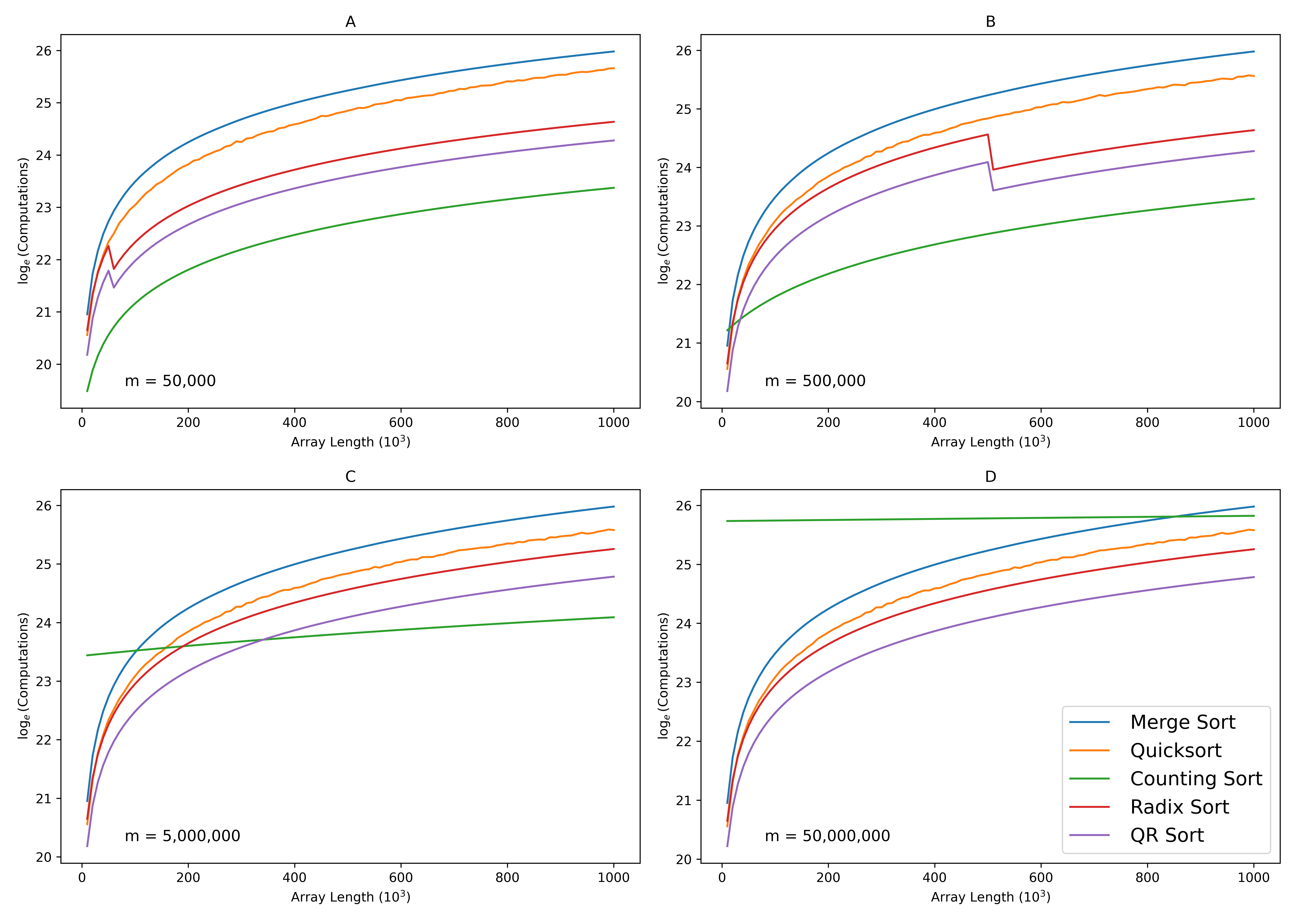}
    \caption{Computational units expended by Merge Sort, Quicksort, Counting Sort, Radix Sort, and QR Sort for four SortTester\_C experiments using input integer arrays $A$ with varying element ranges $m$ ($m = max(A) - min(A) + 1$). For Radix Sort, we set the base of the number system $b$ equals the input array length $n$. (a) $m = 50,000$. QR Sort expended fewer computational units than Merge Sort, Quicksort, and Radix Sort. Counting Sort expended the least computational units compared to all other sorting algorithms. (b) $m = 500,000$. QR Sort expended fewer computational units than Merge Sort, Quicksort, and Radix Sort. Counting Sort generally expended the least computational units compared to all other sorting algorithms. (c) $m = 5,000,000$. QR Sort expended fewer computational units than Merge Sort, Quicksort, and Radix Sort. Counting Sort expended the least computational units compared to all other sorting algorithms when $n \geq 370,000$. QR Sort expended the least computational units compared to all other sorting algorithms when $n < 370,000$. (d) $m = 50,000,000$. QR Sort expended fewer computational units than all other sorting algorithms.}
    \label{graph_figure}
    \Description{}
\end{figure*}

\section{Conclusion}\label{conclusion_section}
In conclusion, we present QR Sort, a novel non-comparative integer sorting algorithm that uses principles derived from the Quotient-Remainder Theorem to sort input sequences stably. We showed that QR Sort outperforms common comparative-based sorting algorithms for a variety of input array lengths and element ranges. We also showed that QR Sort excels in sorting arrays with large array element ranges, and outperforms common non-comparative sorting algorithms in those scenarios. As such, QR Sort could be used in the future to improve the computational efficiency of applications that encounter datasets with large value ranges such as prioritization tasks, graph theory, and database platforms.

\bibliographystyle{ACM-Reference-Format}
\bibliography{bib}

\clearpage

\appendix
\section{Pseudo-Code}\label{appendix_pseudo-code}
\subsection{Counting Sort}\label{appendix_pseudo-code_counting_sort}
\begin{algorithm*}
\caption{Sort Integer Array With Counting Sort Using Provided Keys}
\label{pseudo_code_counting_key_sort}
\begin{algorithmic}[1]

\Function{Counting-Sort}{$A, B, C, K, \text{copy\_B\_to\_A}$}
\State // $A$: Integer array to be sorted of size $n$ (unchanged if copy\_B\_to\_A is false)
\State // $B$: Auxiliary integer array, always contains the sorted version of $A$ after sort
\State // $C$: Counting array, initialized to 0, with size equal to the range of $K$ plus 1
\State // $K$: Integer array that maps $A$[$i$] to the key $K$[$i$] 
\State // copy\_B\_to\_A: If true, copies sorted $B$ back to $A$
\\

\State // $C[k]$ tracks the occurrence frequency of each value $k$ in $K$
\For{$i \gets 0$ to $\text{len(}A\text{)} - 1$} \label{alg:line:occ}
    \State $k \gets K[i]$
    \State $C[k] \gets C[k] + 1$ 
\EndFor
\\

\State // $C[k]$ is updated to track the cumulative frequency of all key values $\leq k$ in $K$
\For{$k \gets 1$ to $\text{len(}C\text{)} - 1$} \label{alg:line:cum_occ}
   \State $C[k] \gets C[k] + C[k-1]$
\EndFor
\\

\State // Use $C$ values to find final $B$ indexes for each $A$ value
\For{$i \gets \text{len(}A\text{)} - 1$ \textbf{down to} $0$} \label{alg:line:find}
    \State $k \gets K[i]$
    \State $B[C[k] - 1] \gets A[i]$
    \State $C[k] \gets C[k] - 1$
\EndFor
\\

\State // Copy $B$ values to $A$ if $\text{copy\_B\_to\_A}$ is true
\If{$\text{copy\_B\_to\_A}$} \label{alg:line:trans}
    \For{$i \gets 0$ to $\text{len(}A\text{)} - 1$} 
        \State $A[i] \gets B[i]$
    \EndFor
\EndIf
\\

\EndFunction
\end{algorithmic}
\bigskip % Adds some vertical space before the complexities section
\noindent\textbf{Computational Complexities:}
\begin{itemize}
    \item Computing occurrence counts: $\mathcal{O}(n)$, see line \ref{alg:line:occ}
    \item Computing cumulative occurrence counts: $\mathcal{O}(\text{len(}C\text{)})$, see line \ref{alg:line:cum_occ}
    \item Locate final output indexes: $\mathcal{O}(n)$, see line \ref{alg:line:find}
    \item Copy auxiliary array to original array: $\mathcal{O}(n)$ if \text{copy\_B\_to\_A} is true; else $\mathcal{O}(1)$, see line \ref{alg:line:trans}
    \item \textbf{Best-Case Computational Complexity:} $\bm{\mathcal{O}(n)}$
    \item \textbf{Worst-Case Computational Complexity:} $\bm{\mathcal{O}(n + \textbf{len(}C\textbf{)})}$
\end{itemize}
\end{algorithm*}

\clearpage
\subsection{QR Sort}\label{appendix_pseudo-code_qr_sort}
\begin{algorithm*}
\caption{Sort Integer Array With QR Sort Using Given Divisor}
\label{psuedo_code_qr_sort}
\begin{algorithmic}[1]

\Function{QR-Sort}{$A, d$}
\State // $A$: Integer array to be sorted of size $n$ 
\State // $d$: A positive integer divisor
\State // Let $m$ denote the range of values in $A$ plus one: $\max(A) - \min(A) + 1$
\\

\State // Compute the maximum normalized quotient in $A$
\State let $\text{max\_quot} \gets \left( \frac{\max(A) - \min(A)}{d} \right)$ \label{alg:line:find_min_max}
\\

\State // Declare an auxiliary array to temporarily hold the values of $A$
\State let $B[0 \ldots \text{len(}A\text{)} - 1]$ be a new integer array
\\

\State // Declare counting array for remainder keys and compute remainder keys
\State let $C_r[0 \ldots d - 1] \gets \{0\}$
\State let $R[0 \ldots \text{len(}A\text{)} - 1] \gets \left\{ (v - \min(A)) \bmod d \mid v \in A \right\}$ \label{alg:line:compute_remainder_keys}
\\

\State // If the maximum quotient equals zero, the quotient sort can be skipped (see Section \ref{which_sort}) 
\If{$\text{max\_quot} == 0$}
    \State // Execute Counting Sort with remainder keys
    \State $\text{COUNTING-SORT}(A, B, C_r, R, \text{true})$ \label{alg:line:rem_sort_1} 

\Else
    \State // Sort A into B using remainder keys while preserving A, avoiding a redundant linear pass
    \State $\text{COUNTING-SORT}(A, B, C_r, R, \text{false})$ \label{alg:line:rem_sort_2}
    \\
    \State // Declare counting array for quotient keys and compute quotient keys
    \State let $C_q[0 \dots \text{max\_quot}] \gets \{0\}$ \label{alg:line:reinit_cnt}
    \State let $Q[0 \ldots \text{len(}B\text{)}-1] \gets \left\{ \left\lfloor \frac{v - \min(B)}{d} \right\rfloor \mid v \in B \right\}$ \label{alg:line:compute_quotient_keys}
    \\
    \State // Sort $B$ into $A$ using quotient keys
    \State $\text{COUNTING-SORT}(B, A, C_q, Q, \text{false})$ \label{alg:line:quot_sort}

\EndIf
\\
\EndFunction
\end{algorithmic}

\bigskip % Adds some vertical space before the complexities section
\noindent\textbf{Computational Complexities:}
\begin{itemize}
    \item Computing maximum quotient value: $\mathcal{O}(n)$, see line \ref{alg:line:find_min_max}
    \item Computing remainders keys: $\mathcal{O}(n)$, see line \ref{alg:line:compute_remainder_keys}
    \item Counting Sort using the remainders keys: $\mathcal{O}(n + d)$, see line \ref{alg:line:rem_sort_1}, \ref{alg:line:rem_sort_2}
    \item Computing quotient keys: $\mathcal{O}(n)$, see line \ref{alg:line:compute_quotient_keys}
    \item Counting Sort using the quotients keys: $\mathcal{O}(n + \frac{m}{d})$, see line \ref{alg:line:quot_sort}
    \item \textbf{Best-Case Computational Complexity:} $\bm{\mathcal{O}(n)}$
    \item \textbf{Worst-Case Computational Complexity:} $\bm{\mathcal{O}(n + d + \frac{m}{d})}$
\end{itemize}
\end{algorithm*}
% \bibliographystyle{unsrt}
% \bibliography{bib.tex}

\end{document}